\documentclass[10pt,onecolumn,draftcls,journal]{IEEEtran}
\usepackage[ansinew]{inputenc}
\usepackage{epsfig}
\usepackage[T1]{fontenc}
\usepackage{amsmath,enumerate,amssymb,hhline,verbatim}
\usepackage{xcolor}

\newenvironment{proof}{\begin{IEEEproof}}{\end{IEEEproof}}

\newcommand{\bbar}[1]{\setbox0=\hbox{$#1$}\dimen0=.2\ht0 \kern\dimen0 \overline{\kern-\dimen0 #1}}

\newtheorem{thm}{Theorem}[section]
\newtheorem{theorem}[thm]{Theorem}

\newtheorem{corollary}[thm]{Corollary}

\newtheorem{lemma}[thm]{Lemma}
\newtheorem{proposition}[thm]{Proposition}
\newtheorem{prop}[thm]{Proposition}

\newtheorem{definition}[thm]{Definition}

\newtheorem{remark}[thm]{Remark}

\newtheorem{example}[thm]{Example}

\newcommand{\Z}{{\mathbb Z}}
\newcommand{\R}{{\mathbb R}}
\newcommand{\C}{{\mathbb C}}
\newcommand{\D}{{\mathcal D}}

\newcommand{\HH}{{\mathbf H}}
\newcommand{\OO}{{\mathcal O}}

\newcommand{\A}{{\mathcal A}}
\newcommand{\B}{{\mathcal B}}

\newcommand{\Q}{{\mathbb Q}}

\newcommand{\abs}[1]{\vert #1 \vert}
\newcommand{\mindet}[1]{\hbox{\rm det}_{min}\left( #1\right)}
\newcommand{\diag}{{\rm diag}}

\begin{document}

\title{A non-commutative analogue of the Odlyzko bounds and bounds on performance for space-time lattice codes 
\thanks{The research of B.~Linowitz was partially supported by an NSF RTG grant DMS-1045119 and an NSF Mathematical Sciences Postdoctoral Fellowship. }
\thanks{The research of   M.~Satriano was partially supported by an NSF Mathematical Sciences Postdoctoral Fellowship.}
\thanks{The research of R.~Vehkalahti is funded by   Academy of Finland  grant  \#252457.}}

\author{Benjamin Linowitz,  Matthew Satriano and Roope Vehkalahti
\thanks{B.~Linowitz is with the Department of Mathematics, University of Michigan, 530 Church Street, Ann Arbor, MI, 48109, USA (email: linowitz@umich.edu)}
\thanks{M.~Satriano is with Division of Biostatistics and Bioinformatics, Department of Oncology, Sidney Kimmel Comprehensive Cancer Center, Johns Hopkins University School of Medicine, 550 North Broadway, Suite 1103, Baltimore, MD 21205 (email: msatria2@jhu.edu)}
\thanks{R.~Vehkalahti is with the Department of Mathematics, FI-20014, University of Turku, Finland  (e-mail: roiive@utu.fi)}}

\maketitle

\begin{abstract}
This paper considers space-time coding over several independently Rayleigh faded blocks. In particular we will concentrate on giving upper bounds for 
the coding gain of lattice space-time codes as the number of blocks grow.  This problem was previously considered in the single antenna case by Bayer et al. in 2006.    Crucial to their work was Odlyzko's bound on the discriminant of an algebraic number field, as this provides an upper bound for the normalized coding gain of number field codes.  In the MIMO context natural codes are constructed from division algebras defined over number fields and the coding gain is measured by the discriminant of the corresponding (non-commutative) algebra. In this paper we will develop analogues of the Odlyzko bounds in this context and show how these bounds limit the normalized coding gain of a very general family of division algebra based  space-time codes. These bounds can also be used as benchmarks in practical code design and as tools to analyze asymptotic bounds of performance as the number of independently faded blocks increases.
\end{abstract}

\section{Introduction}


Consider a lattice $L \subset \C^n$ having fundamental parallelotope of volume one and define a function
$f_1:\C^n\to \R$ by 
\begin{equation}\label{euclidean}
f_1(x_1,\dots,x_n)=|x_1|^2+|x_2|^2+\cdots+ |x_n|^2.
\end{equation} 
The real number $h(L)=\inf_{x \in L, \, x\neq{\bf 0}} f_1(x)$ is the \emph{Hermite invariant} of the lattice $L$. In rough terms we may say that the greater the Hermite invariant of a lattice is, the higher the guaranteed protection against worst case pairwise error when a subset of the lattice is used as a code in the Gaussian or slow fading channel.
Similarly, if we have a Rayleigh fast fading single antenna channel, the role of the function $f_1$ is played by the function
\begin{equation}\label{fastfading}
f_2(x_1,x_2,\dots, x_n)=|x_1x_2\cdots x_n|.
\end{equation}
The real number $Nd_{p, min}(L)=\inf_{x\in L, x\neq{\bf 0}} f_2(x) $ is the \emph{normalized product distance} of the lattice $L$ and can be used to identify the best lattice code for the fast fading channel on high signal-to-noise ratio (SNR) regime.

Let us now consider the  main topic of this paper. Suppose that we  have $n$ transmit antennas and a \emph{Rayleigh block fading channel} where the  fading stays stable for $n$ units of  time and then changes independently for the next $n$ units of time. The ability to encode and decode over $m$ such independently faded blocks implies that our lattice code $L$ lies in the space $M_{n\times mn}(\C)$. Let us suppose that $(X_1, X_2,\dots, X_m)$  is an element of $M_{n\times mn}(\C)$,  and define
\begin{equation}\label{mindet}
f_3(X_1, X_2,\dots, X_m)=\prod_{i=1}^m |det(X_i)|.
\end{equation}
In analogy with the functions defined above, we can define the \emph{normalized minimum determinant} of the lattice $L$ by
$\delta(L)=\inf_{X\in L, X\neq{\bf 0}}|f_3(X)| $. Again, the number $\delta(L)$ can be seen as measuring the quality of the lattice L.

The following problems are natural to consider in all three cases. 
\begin{itemize}
\item [1]  Given the channel, find optimal lattices that maximize the corresponding function $f_i$.
\item [2]  Given a lattice, find upper and lower bounds for the maximal value obtained by the function $f_i$.
\end{itemize}

From a mathematical standpoint these problems can be seen as arising in the classical \emph{geometry of numbers}, though good solutions for the problems in full generality do not appear to be in the literature. 

The case in which one considers the function $f_1$, defined in \eqref{euclidean}, and the associated Hermite invariant, is known as the sphere packing problem. In this setting there exist a number of good constructions and general bounds.  

In the case of the function $f_2$, defined in \eqref{fastfading}, and the associated real number $Nd_{p, min}(L)$, most of the known constructions are based on algebraic number fields and good general  bounds are known  only in the case where the lattice $L$ is real \cite[p. 615]{Hasse}.  

For the function $f_3$ defined in (\ref{mindet}) and the associated real number $\delta(L)$, to the best of our knowledge there are no good general bounds.

For a general lattice $L\subset M_{n\times mn}(\C)$ finding good bounds for $\delta(L)$ is an extremely difficult task. For this reason we restrict our attention to a broad class of lattices arising from central division algebras defined over number fields. For the lattices arising from this construction we can say a great deal more about $\delta(L)$.

In order to describe our results on  $\delta(L)$, let us first briefly describe the general construction principle behind these algebraic lattices. There are many ways to construct lattice codes with good Hermite invariant.  To build a lattice code with good product distance or minimum determinant, the task is more difficult.  A usual method is to choose a central simple algebra or number field $\A$, a suitable subset $\Lambda\subset \A$ and then a faithful representation $\psi$ which maps every element of $\A$ to a suitable matrix space $M_{m\times mn}(\C)$. If the mapping $\psi$, the subset $\Lambda$, and the algebra $\A$ are well chosen then the set  $\psi(\Lambda)$ will be a lattice in $M_{n\times mn}(\C)$  and will have a  good minimum determinant. This type of construction offers a rich selection of lattice codes. 
Assuming that this algebraic construction yields a $k$-dimensional lattice in the given matrix space $M_{n\times mn}(\C)$, it is natural to ask for bounds on the size of the minimum determinant.

This problem was first considered in the context of number field codes in a fast fading SISO channel in \cite{FOV} by using the Odlyzko bounds \cite{Odlyzko-bounds}, and in \cite{Xing} by using sphere packing bounds. In  \cite{BORV}, \cite{HLRV2} and \cite{RoopeThesis} the problem was considered in the case in which $m=1$ and $n\geq 1$. In \cite{HL} the authors concentrated on the case where $n=m=2$.

In this work we will generalize and unify previous  number field and division algebra constructions and relate the normalized  minimum determinant to the \emph{discriminant} of the corresponding algebra. We will then give completely general lower bounds for the discriminant of any division algebra and derive upper bounds for the minimum determinants of the corresponding lattices. The discriminant bounds given in this paper are a generalization of the Odlyzko bounds in number fields and are of independent interest.

\smallskip
 
We will begin by defining the channel model, lattice codes and finite codes associated to a lattice. We will then describe the suitability of the normalized minimum determinant as a design criterion in one shot MIMO channels and make some remarks on the limits of this criterion.
In Section \ref{multiblockdesign} we show how this criterion  can  be extended to the multi-block channel.  In  Section \ref{oneshotalgebra}   we briefly review the known construction methods of lattice codes from division algebras. We then extend these methods so as to obtain a lattice code for a multiblock channel from any order in a central division algebra. The presented explicit methods follow \cite{YB07} and \cite{Lu}, and unify \cite{HL} and \cite{VHO}. The construction method given in Proposition \ref{abs2} generalizes the previously used methods by allowing us to consider a larger array of centers.

 Section \ref{bounding_sec}  contains the main results of our paper. In   construction Sections \ref{oneshotalgebra}  and \ref{construction} we did prove that in most cases the normalized minimum determinant of a division algebra code depends only on the discriminant of the algebra.  Unlike the case of number fields however, the mathematical literature does not offer ready-to-use bounds for the discriminant of a central division algebra defined over a number field. The discriminant bounds in \cite{RoopeThesis} do solve this problem, but only after the center is fixed.  However, in the general case, where we are allowed to optimize our code over all number fields with a fixed degree (or even signature), these results do not apply. The problem is that the $\Z$-discriminant of a central division algebra $\A$  defined over a number field $K$ is   a product of terms depending  on the discriminant of the center  $d(\OO_K/\Z)$ and the $K$-discriminant $d(\Lambda/ \OO_K)$ of  the algebra $\A$ and minimizing  one term might implicitly make the other term bigger.  This problem was first considered in \cite{HL}, where the authors were able to solve the problem for division algebras of degree $2$ over totally complex fields of degree $4$.

In Section \ref{bounding_sec}  we will give completely general lower bounds which make no assumptions on the degree (or even signature) of the center or division algebra. The proofs of these results combine the bounds in \cite{RoopeThesis}  with an analysis of the proof of the original Odlyzko bounds for number fields. As described in \cite{RoopeThesis}, the discriminant $d(\Lambda/ \OO_K)$  depends only on the two prime ideals of $\OO_K$ of smallest norm. In order to find lower bounds we make crucial use of the fact that the original proof of Odlyzko (and certainly its refinement due to Poitou \cite{Poitou}) describes the impact on $d(\OO_K/Z)$ of the assumption that the field $K$ has prime ideals with small norm. The proofs of our  theorems are a result of balancing the effect of  small primes on $d(\Lambda/\OO_K)$ against their effect on making the field discriminant $d(\OO_K/\Z)$  bigger. The bounds we develop therefore form a non-commutative analogue of the Odlyzko bounds in algebraic number theory.

While Section \ref{bounding_sec} is strongly mathematical, Section \ref{examples} returns to the coding theoretic context. We first derive some easy-to-use corollaries of our main theorems and then show how it is possible to find algebras that are optimal for our bounds. We then show how these discriminant bounds can be used so as to deduce minimum determinant bounds.  Finally, we compare the resulting bounds to the minimum determinants of some example codes.

As in the case of the traditional Odlyzko bounds, our non-commutative bounds seem to be very tight. Therefore the bounds  can be used as a benchmark in code design and provide understanding of the asymptotic behavior of the  worst case pairwise error probability.  The weakness of our approach 
lies in  the fact that while minimum determinant criteria (in one form or another) has been applied in a number of space-time coding papers it only considers pairwise error and not the actual error probability. We will discuss this issue in Section \ref{oneshotalgebra} and suggest a remedy  to this problem.

\subsection{Channel model}

In this paper we are considering the so called multiblock Rayleigh faded channel with minimal delay. In such a channel a codeword  
$X\in M_{n\times nm}(\C)$ has the form $(X_1,X_2,\dots, X_m)$, where $X_i \in M_n(\C)$. The channel equation, for transmitting i'th block $X_i$, then has the form
\begin{equation}\label{eq:channel}
Y_i=H_i X_i +N_i, 
\end{equation}
where $H_i \in M_{n_r \times n}(\C)$ is the channel matrix and $N_i\in M_{n_r \times T}(\C)$ is the noise; $n_r$ denotes the number of receiving antennas. Here we assume that each of $H_i$ are independently Rayleigh faded and the decoding is  done after the receiver has received all $m$ blocks.  We will call such a channel an $(n, n_r, m)$-multiblock channel. We note that when $m=1$ this is the usual one shot MIMO channel and when $n=1$ we are dealing with the fast fading single antenna channel.

A code  $C$ in  a  $(n, n_r, m)$-channel is  a set of matrices  in $M_{n\times nm}(\C)$. This paper will discuss  code design and performance limits of codes in this channel. In particular we will concentrate on finite codes that are drawn from lattices and we assume that the receiver has  perfect channel state information.

\subsection{Lattices and spherical shaping }

\begin{definition}
A {\em matrix lattice} $L \subseteq M_{n\times T}(\C)$ has the form
$$
L=\Z B_1\oplus \Z B_2\oplus \cdots \oplus \Z B_k,
$$
where the matrices $B_1,\dots, B_k$ are linearly independent over $\R$, i.e., form a lattice basis, and $k$ is
called the \emph{rank}  or the \emph{dimension} of the lattice.
\end{definition}

The space $M_{n\times T}(\C)$ is a $2nT$-dimensional real vector space with a real inner product
$$
\langle X,Y\rangle=\Re(Tr (XY^{\dagger})),
$$
where $Tr$ is the matrix trace. This inner product also naturally defines a metric on the space $M_{n\times T}(\C)$ by setting $||X||_F= \sqrt{\langle X,X\rangle}$.

We now consider a spherical shaping scheme based on a $k$-dimensional lattice $L$ inside $M_{n\times T}(\C)$. 
Given a positive real number $R$ we define 
$$
L(R)=\{X \in L : ||X||_F \leq R, X\neq {\bf 0} \}.
$$

These codes $L(R)$ will be  the finite codes we are considering.

\subsection{Design criterion for one shot MIMO}\label{oneshot}
Before  presenting a design criterion for the multiblock channel we describe the minimum determinant criterion used in the usual MIMO Rayleigh fading channel.
The concept of  \emph{normalized minimum determinant} we are going to define has appeared implicitly or in restricted forms in several papers in space-time coding.   Early attempts to define it in generality were given in \cite{WX} and \cite{WWX}, but only in \cite{LV} was the normalized minimum determinant defined formally and in a manner completely analogous to the definition of the Hermite invariant. Despite various papers where it has been used as a code design criterion, the needed energy normalizations still seem to cause confusion. We will therefore  try to give an improved explanation of the concept here.

Let us suppose that $L$ is a $k$-dimensional lattice in  $M_{n\times T}(\C)$ and that we consider a finite code $L(R)$.
Let $\theta$ be a positive constant with the property that
$$
\frac{1}{|L(R)|}\sum_{X\in L(R)} ||\theta X||^2=T.
$$

Let us now consider transmission of codewords from $L(R)$ in  the Rayleigh fading MIMO channel with $n=n_t$ transmit antennas and $n_r$ receive antennas. The channel is assumed to be fixed for a block of $T$ channel uses, but to vary in an independent and identically distributed (i.i.d.) fashion from one block to another. Thus, the channel input-output relation can be written as 
\begin{equation}
Y=\sqrt{\rho}H\theta X +N, \label{eq:channel2}
\end{equation}
where $H \in M_{n_r \times n}(\C)$ is the channel matrix and $N\in M_{n_r \times T}(\C) $ is the noise matrix. The entries of $H$ and $N$ are assumed to be  i.i.d. zero-mean complex circular symmetric Gaussian random variables with variance 1.  
The matrix $X \in L(R)$ is the transmitted codeword, and the term $\rho$ denotes the signal-to-noise ratio (SNR).

Following \cite{TSC},  we can bound the pairwise error probability between two codewords $X \neq X' \in L(R)$ by above when transmitting with SNR $\rho$:
\begin{equation}\label{probsum}
P(\rho,X \to X')\leq \frac{1}{(\det(I+ \frac{\rho\theta^2}{4n}(X-X')(X-X')^*))^{n_r}}, 
\end{equation}
where $*$ denotes complex conjugate transpose.

Combining this expression with the union bound we can now deduce an upper bound for the average error probability when transmitting a codeword from $L(R)$ at SNR $\rho$:
$$
P_e(\rho)\leq  \sum_{\substack{X\in L,\\ 0 <||X||_F \leq 2R}}\frac{1}{(\det(I+ \frac{\rho\theta^2}{4n}XX^*))^{n_r}}.
$$
If we suppose that $\rho$ is particularly large and the matrices $X$ are invertible then we obtain the further bound
\begin{equation}\label{inversesum}
P_e(\rho)\leq  \sum_{\substack{X\in L,\\ 0 <||X||_F \leq 2R}}\frac{1}{(\det(\frac{\rho\theta^2}{4n} XX^*))^{n_r}}.
\end{equation}
In what follows, the matrices in our lattices will not only be invertible but have an even stronger property.

\begin{definition}\label{def:NVD}
If the \emph{minimum determinant} of the lattice $L \subseteq M_{n\times T}(\C)$ is non-zero, i.e. satisfies
\[
\mindet{L}:=\inf_{{\bf 0} \neq X \in L} \sqrt{\abs{\det (XX^*)}} > 0, 
\]
we say that the lattice satisfies the \emph{non-vanishing determinant} (NVD) property.
\end{definition}

Assuming now that $\mindet{L}=c$,  we can further improve our inequality \eqref{inversesum} with
\begin{equation}\label{mindetsum}
P_e(\rho)\leq  \sum_{\substack{X\in L,\\ 0 <||X||_F \leq 2R}}\frac{(4n)^{nn_r}}{c^{2n_r}\theta^{2nn_r}\rho^{nn_r}}.
\end{equation}
This bound suggests that  minimum determinant plays a crucial role in the code design. However, in order to compare two $k$-dimensional lattices $L_1$ and $L_2$, the comparison based on the minimum determinant is relevant only if both the needed constants $\theta_1$, $\theta_2$   and
 the number of codewords in $L_1(R)$ and $L_2(R)$ are close to one another.  Therefore we need a normalization that guarantees a fair comparison.

Let us now suppose we have a  $k$ dimensional lattice $L\subset M_{n\times T}(\C)$. 
The   {\it Gram matrix} of the lattice $L$ is defined as 
$$G(L)=\left( \langle X_i,X_j\rangle)\right)_{1\le i,j\le k},$$
where $\{ X_i \}$ is a basis of $L$.  The volume of the fundamental parallelotope of $L$ is then defined as 
$vol(L)=\sqrt{|det(G(L))|}$.

The following lemma proves that if we have two  $k$-dimensional lattices  $L_1$ and $L_2$ in the same space $M_{n\times T}(\C)$, then the scaling factors $\theta_1$ and $\theta_2$ needed for normalization are roughly the same  and the finite codes $L_1(R)$ and $L_2(R)$ have roughly the same number of codewords.

Although both of the assertions in the following lemma are well known, we give a complete proof of the second as it seems to have caused some confusion within the space-time community.
\begin{lemma}\label{spherical}
Let $L$ be a  $k$-dimensional lattice  with a unit fundamental parallelotope in  $M_{n\times T}(\C)$ and
$L(R)$ be defined as above. Then
$$
|L(R)|= c_1M^{k}+ O(R^{k-1})\,\, \mathrm{and} \, \sum_{X\in L(R)} ||X||_F^2 = c_2R^{k+2}+ O(R^{k+1}),
$$
where $c_i$ are constants independent of $R$ and the lattice $L$, and $O$ is Landau's big $O$.
\end{lemma}

\begin{proof}
Let us denote the Voronoi cell of a point $x\in L$ by $V_x$, and let $r$ be a real number such that for any $x\in L$ and for any $y\in V_x$ we have $||x-y||<r$. We then have that for any $x\in L$ and for any $y\in V_x$,  
$$
||x||^2 -2r||x||-r^2\leq ||y||^2\leq ||x||^2+ 2r||x||+r^2. 
$$
 We can therefore write
$$
\sum_{x\in L(R-r)} (||x||^2-2r||x||-r^2)vol(V_x) \leq \int_{B(R)}  ||x||^2 dx \leq  \sum_{x\in L(R)} (||x||^2+2r||x||+r^2)vol(V_x),
$$
where $B(R)$ is a closed ball of radius $M$ about the origin. As we have assumed that $vol(L)=1$, we have that $vol(V_x)=1$ for all $x$ and
\begin{equation}\label{estimate}
\sum_{x\in L(R-r)} ||x||^2 -\sum_{x\in L(R-r)}(2r||x|| +r^2)\leq  \int_{B(R)} ||x||^2 dx \leq  \sum_{x\in L(R)} ||x||^2 +  \sum_{x\in L(R)}(2r||x||+r^2).
\end{equation}
 The integral in the middle grows like $cR^{k+2} +O(R^{k+1})$ and according to the first statement the sum $\sum_{x\in L(R)}  2r||x||+ r^2$ is bounded above by  $CR^{k+1}$ for some $C$ independent of $R$. Using  again the first statement we have that  $\sum_{x\in L(R)} ||x||^2-\sum_{x\in L(R-r)} ||x||^2\in O(R^{k+1})$. Taking all these into account and reorganizing \eqref{estimate} we get the claim.
\end{proof}

We can now define the {\em normalized minimum determinant} $\delta(L)$, which  is obtained  by first scaling the lattice $L$ to have a unit size fundamental parallelotope and then taking the minimum determinant of the resulting scaled lattice. A simple computation proves the following.
\begin{lemma}\label{scale}
Let $L$ be a $k$-dimensional matrix lattice
in $M_{n\times T}(\C)$. We then have that
$$
\delta(L) =\mindet{L}/(vol(L))^{n/k}.
$$
\end{lemma}

\begin{remark}
Different forms of minimum determinant criteria have been used in numerous papers on space-time coding. While a crude tool, the concept has been quite effective in code design. However, the derivation of the minimum determinant criterion through the union bound as in Lemma \ref{spherical} makes it clear  that the \emph{distribution} of the determinants in the lattice, and not just the minimum determinant, is quite relevant. This is particularly clear when the SNR compared to the code size is relatively small. This was already known in the very early work on number field codes \cite{BERB}, though the technical obstacles needed to analyze the question did not allow researchers at the time to attack this problem.

In the context of algebraic codes this problem was addressed in \cite{VLL2013}, where the distribution of determinants of number field and division algebra codes was analyzed.  This work revealed that division algebra based codes can be divided into different classes with respect to their \emph{signature} (defined below in Definition \ref{signature}).  However, as pointed out in \cite{LV2013}, the normalized minimum  determinant still plays a major role and is effective when we compare codes having the same signature.  The bounds we  will develop in this paper are sensitive to the signatures of the considered algebras, and can therefore be used to analyze the behavior of minimum determinants within the class of algebras having the same signature. We can conclude that while minimum determinant criteria are not a perfect measure of space-time codes, the bounds presented here will also be relevant to more refined analyses.
\end{remark}

\subsection{Design criterion for  multiblock  channel}\label{multiblockdesign}
Let us now show how the design criterion of the previous section can be used to define a design criterion for the multiblock channel. Let us  suppose we have a multiblock code $L\subset M_{n\times mn}(\C)$ and that $(X_1,X_2,\dots X_m)$ is a codeword in $L$.
The channel equation 
$$
(H_1X_1, H_2X_2,\dots, H_m X_m) +(N_1,N_2,\dots, N_m),
$$
can  just as well be written in the form
$$
(H_1,H_2,\cdots, H_m)\diag(X_1,X_2,\cdots, X_m) +\diag(N_1, N_2,\dots, N_m),
$$
where the diag-operator places the $i$th $n\times n$ entry in the $i$th diagonal block of a matrix
in $M_{mn}(\C)$. This reveals that  optimizing a code $L$ for the $(n, n_r, m)$-multiblock channel is equivalent to optimizing
$\diag(L)$ for the usual one shot $nm\times mn_r$ MIMO channel, where $\diag(L)$ is defined as $\{\diag(X)\mid X\in L\}$.

Let us now suppose we have  an $(n, n_r, m)$-multiblock code $L$.
\begin{definition}
By abusing notation we define the normalized minimum determinant for the code $L$ by
$$
\delta(L):=\delta(\diag(L)).
$$
\end{definition}

\begin{remark}
This definition is just a generalization of the \emph{minimum product distance} used in  fast fading  single antenna channels.
\end{remark}

We are now interested in the extrema of the normalized minimum determinants of $k$-dimensional  $(n, n_r, m)$-multiblock codes.

\section{Algebraic preliminaries and lattice codes for one shot MIMO }\label{oneshotalgebra}
Let us now describe how lattice codes from division algebras are typically built. We will follow the standard presentation (see \cite{OBV}), but with an order-theoretic perspective \cite{HLL}. The general idea is to  show how we can transform an abstract algebraic structure into a concrete lattice of matrices. This construction will form a basis for our  construction of multiblock codes. We refer the reader to \cite{OBV} for all proofs.

\begin{definition}\label{cyclic}
Let $K$ be an algebraic number field of degree $m$ and assume   that $E/K$ is a cyclic Galois
 extension of degree $n$ with Galois group
$Gal(E/K)=\left\langle \sigma\right\rangle$. We can define an associative $K$-algebra
$$
\mathcal{A}=(E/K,\sigma,\gamma)=E\oplus uE\oplus u^2E\oplus\cdots\oplus u^{n-1}E,
$$
where   $u\in\mathcal{A}$ is an auxiliary
generating element subject to the relations
$xu=u\sigma(x)$ for all $x\in E$ and $u^n=\gamma\in K^*$. We call the resulting algebra a \emph{cyclic algebra}.
\end{definition}

It is clear that the center of the algebra $\mathcal{A}$ is precisely the field $K$. That is, an element of $\A$ commutes with all other elements of $\A$ if and only if it lies in $K$.

 \begin{definition}
 We  call $\sqrt{[\A:K]}$ the \emph{degree} of the algebra $\A$. It is easily verified that the degree of $\A$ is equal to $n$.
\end{definition}

We consider $\A$ as a right  vector space over $E$
and note that every  element $a=x_0+ux_1+\cdots+u^{n-1}x_{n-1}\in\mathcal{A}$
has the following representation as a matrix
\[\psi(a)=A=\begin{pmatrix}
x_0& \gamma\sigma(x_{n-1})& \gamma\sigma^2(x_{n-2})&\cdots &
\gamma\sigma^{n-1}(x_1)\\
x_1&\sigma(x_0)&\gamma\sigma^2(x_{n-1})& &\gamma\sigma^{n-1}(x_2)\\
x_2& \sigma(x_1)&\sigma^2(x_0)& &\gamma\sigma^{n-1}(x_3)\\
\vdots& & & & \vdots\\
x_{n-1}& \sigma(x_{n-2})&\sigma^2(x_{n-3})&\cdots&\sigma^{n-1}(x_0)\\
\end{pmatrix}.
\]

This mapping allows us to embed any cyclic algebra into $M_n(\C)$. Under such an embedding $\psi(\A)$ forms an $mn^2$-dimensional $\Q$-vector space.  The map $\psi$ is called the \emph{left regular representation} of $\A$ (see Remark \ref{regremark}).

We are particularly interested in algebras $\A$ for which $\psi(a)$ is invertible for all non-zero $a\in\A$.

\begin{definition}\label{divisionalgebra}
A cyclic $K$-algebra $\A$ is a \emph{division algebra} if every
non-zero element of $\A$ is invertible.
\end{definition}

The set  $\psi(\A)$  is an additive subgroup of $M_n(\C)$ but is not discrete. This is obviously not a preferred property for a lattice code.
A usual strategy to try to overcome this problem is to restrict one's attention to the image in $M_n(\C)$ of a suitable subset of $\A$.

\begin{definition}
 A \emph{$\Z$-order} $\Lambda$ in $\A$ is a subring of $\A$ having the same identity element as
$\A$, and such that $\Lambda$ is a finitely generated
module over $\Z$ which generates $\mathcal{A}$ as a linear space over $\Q$.
\end{definition}

\begin{lemma}
Let $\Lambda$ be a $\Z$-order in a division algebra $\A$. We then  have that  $\psi(\Lambda)$ is a free group with $mn^2$ generators. 
In other words
$$
\psi(\Lambda)=\Z B_1\oplus \Z B_2\oplus \cdots \oplus \Z B_{mn^2}\subset M_n(\C).
$$
We also have that $det(X)\neq 0$ for every non-zero element $X\in\psi(\Lambda)$.
\end{lemma}

Although $\psi(\Lambda)$ is  an additive group, it is not usually a lattice; indeed, if $m>2$ then the matrices $\B_i$ cannot be linearly independent over $\R$, as $mn^2>2n^2$. Lattice theory then tells us that $\psi(\Lambda)$ is not a discrete set under such conditions.


It can be proven that if $K$ is either $\Q$ or a complex quadratic field, then $\psi(\Lambda)$  is  a lattice in $M_n(\C)$ and will have the NVD property. For other division algebras, as we pointed out above, $\psi(\Lambda)$  is not a lattice in $M_n(\C)$. However, this does not exclude the possibility that there is a \emph{different} embedding $\psi'$ of $\A$ into a matrix space that realizes $\Lambda$ as an NVD lattice.

Algebraic existence results (particularly the short exact sequence of Brauer groups that appears in local class field theory \cite[Eq.32.13]{Reiner} show that, given an algebraic number field  $K$ of degree $m$ and any integer $n\geq 1$, there exist infinitely many isomorphism classes of central division algebras of degree $n$ having center equal to $K$. Furthermore, the Albert-Brauer-Hasse-Noether theorem implies that every central simple algebra defined over a number field is cyclic \cite[Theorem 32.20]{Reiner}, and therefore of the form given in Definition \ref{cyclic}. We will show in the following sections that for every order $\Lambda$ in a division algebra $\A$, there is an embedding $\psi'$ of $\A$ into a suitable matrix such that the resulting code $\psi'(\Lambda)$ is a multiblock code with the NVD property. 


\begin{remark}\label{regremark}
In order to state our constructions in Section \ref{construction} in full generality, we need a more general version of the left regular representation.
Let $\A$ be a central division algebra of degree $n$ over a number field $K$. Let us now suppose that $E$ is a maximal subfield of $\A$. From the theory of central simple algebras, we know that $[\A:E]=n$. Let $\{ d_1,\dots, d_n \}$ be a right $E$-basis for $\A$.  Multiplication on the left is an $E$-linear mapping of $\A$ into itself. In this manner we get a $K$-algebra embedding
$\phi:\A\hookrightarrow M_n(E) \subseteq M_n(\C)$. We call this embedding the \emph{left regular representation}. 
\end{remark}



Lastly, given a division algebra $\A$ over a number field, to every $\Z$-order $\Lambda$ in $\A$, we can associate a non-zero integer $d(\Lambda/\Z)$ called the \emph{$\Z$-discriminant of $\Lambda$}. Although we do not give the definition here, throughout the paper we give references for all propeties of $\Z$-discriminants that we use. We refer the reader to \cite[Chapter 2]{Reiner} for a detailed treatment of the theory of orders in central simple algebras.

\section{Multiblock codes from central division algebras}\label{construction}
In this section we will describe how we can build multiblock lattice codes from division algebras and how it is possible to  measure the normalized minimum determinants of the constructed codes in terms of algebraic invariants of the corresponding division algebras. The main theme here is that we begin with an ``idealized'' abstract embedding $\psi_{abs}$ that gives us an existence result where any order $\Lambda$ of a division algebra $\A$ can be realized as a multiblock  lattice code $\psi(\Lambda)\subset M_{n\times nk}(\C)$ having the NVD property. The normalized minimum determinant of the corresponding code is directly related to the discriminant of the order $\Lambda$.
We then try to find an explicit embedding $\psi_{reg}$ that has many of the same properties of the abstract embedding and for which the connection between the discriminant and the minimum determinant still holds. Our presentation follows \cite{VHO}. Only in Section  \ref{highdegree} will we extend beyond \cite{VHO}.

We begin with a few definitions and preliminary results.

Let $K/\Q$ be an algebraic number field of degree $m$. We then have that
\[
m = r_1 + 2r_2,
\]
where $r_1$ is the number of real embeddings and $r_2$ the number of pairs of complex embeddings
of $K$ into $\C$.

Let us define the space $G(\C)_{n}\subseteq M_{n\times2n}(\C)$ by
\[
G(\C)_{n}=\{ (\bbar{B},B)\in M_{n\times2n}(\C) : B\in M_n(\C)\},
\]
where $^*$ refers to complex conjugation and $\bbar{B}=(b_{ij}^*)$.

\begin{definition}
\label{def:quaternions}
The ring $\HH$ of \emph{Hamiltonian quaternions} is a subset in $M_2(\C)$ consisting of matrices of type 
$$
\begin{pmatrix}
x_1 & -x_2^* \\
x_2 & x_1^*
\end{pmatrix},
$$
where $x_i \in \C$ are freely chosen.  Each matrix in the matrix ring $M_{n}(\HH) \subset M_{2n}(\C)$ consists of
$n^2$ freely chosen $(2\times 2)$ blocks that have the inner structure of Hamiltonian quaternions.
\end{definition}

There exists an isomorphism (see \cite{BCC})
\begin{equation}\label{mainmap}
\A\otimes_\Q\R \cong M_{n/2}(\HH)^{\omega} \times M_{n}(\R)^{r_1 -\omega} \times G(\C)_n^{r_2}.
\end{equation}
The integer $\omega$ appearing in (\ref{mainmap}) is, by definition, the number of real places where $\A$ ramifies. 

\begin{definition}\label{signature}
We call the triplet $(\omega, r_1, r_2)$  the \emph{signature} of $\A$.
\end{definition}

Each element in $\A$ can now 
be seen as a concatenation of $\omega$ matrices in $M_n(\C)$, $r_1-\omega$ matrices in $M_n(\R)$ and
$r_2$ pairs of conjugate matrices in $M_n(\C)$. Equivalently, every element of $\A$ can be viewed as a matrix in $M_{n\times nm}(\C)$, where as above $m=r_1+2r_2$.

The above isomorphism (\ref{mainmap}) implies the existence of an injection $\psi_{abs}$
\begin{equation}\label{absmap}
\A \hookrightarrow (M_{n/2}(\HH)^{\omega} \times M_{n}(\R)^{r_1 -\omega} \times G(\C)_n^{r_2}) \subset M_{n\times nm}(\C).
\end{equation}
This will be our "idealized" abstract embedding.

\subsection{Division algebra based $mn^2$-dimensional codes  in $M_{n\times nm}(\C)$}\label{lowdegree}
We will now finally show how any order inside of an arbitrary central division algebra $\A$ can be realized as a lattice in a suitable matrix space.

Let us first describe the codes and their properties we get by using the embedding \eqref{absmap}.

Let $K/\Q$ be a number field of degree $m$ and $\A$ a $K$-central division algebra of degree $n$. 

\begin{proposition}[\cite{VHO}]\label{abs1}
Let us suppose that $\Lambda$ is a $\Z$-order in $\A$ and $\psi_{abs}$ the embedding \eqref{absmap}.
Then $\psi_{abs}(\Lambda)$ is an $n^2m$-dimensional lattice in $M_{n\times nm}(\C)$ and
$$
\mindet{\psi_{abs}(\Lambda)}= 1, \,\,\, vol(\psi_{abs}(\Lambda))=\sqrt{d(\Lambda/\Z)}, \,\,\mathrm{ and }\,\,\delta({\psi_{abs}(\Lambda)})=\left(\frac{1}{|d(\Lambda/\Z)|}\right)^{1/2n}.
$$
\end{proposition}

This result gives us the existence result. We now know that any order of a division algebra can be realized as a multiblock code. However, the embedding \eqref{absmap} is based on existence results and does not directly give us a method to find the lattices of Proposition \ref{abs1}. Yet it does give us a hint of how it can be imitated in an explicit way

Let $K$ and $\A$ be as above, $E$ be a maximal subfield of $\A$ and $\phi:\A\hookrightarrow M_n(E) \subseteq M_n(\C)$ the left regular representation.

The field $K$ has $m$ distinct $\Q$-embeddings $\beta_i$ from $K$ into $\C$. For each $\beta_i$ we can find an embedding $\alpha_i: E\hookrightarrow \C$ which extends $\beta_i$ in the sense that $\alpha_i|_{K}=\beta_i$. We caution the reader that the embedding $\alpha_i$ will not in general be unique. Let us now suppose that $\{\alpha_1,\dots, \alpha_m\}$  is collection of embeddings of $E$ into $\C$ which extend all of the embeddings $\{\beta_1,\dots,\beta_m\}$. Let $a$ be an element of $\A$ and $A=\phi(a)$ the corresponding matrix in $M_n(E)$. We then get a mapping $\psi_{reg1}: \A\to M_{n\times nm}(\C)$ given by
\begin{equation}\label{regmap1}
d\mapsto (\alpha_1(A),\dots, \alpha_m(A)),
\end{equation}
where each of the embeddings $\alpha_i$  have been extended to  maps $\alpha_i: M_n(E)\hookrightarrow M_n(\C)$.

\begin{prop}[\cite{VHO}]\label{reg1}
Let $\Lambda$ be a $\Z$-order in $\A$ and $\psi_{reg1}$ the previously defined embedding.
Then $\psi_{reg1}(\Lambda)$ is an $n^2m$-dimensional lattice in $M_{n\times nm}(\C)$ and $\mindet{\psi_{reg1}(\Lambda)}= 1$.
\end{prop}

We are now interested in the values of  $\delta(\psi_{reg1}(\Lambda))$. As we know that $\mindet{\psi_{reg1}(\Lambda)}= 1$,  Lemma \ref{scale} implies that in order to measure $\delta(\psi_{reg1}(\Lambda))$ it suffices to know $vol(\psi_{reg1}(\Lambda))$. Unfortunately we cannot always relate this value to the algebraic invariants of $\A$.  The following result describes conditions under which we can determine the normalized minimum determinant of the code from the discriminant of the associated order.

\begin{proposition}\label{equal}
Let us suppose that $\A$ has signature $(\omega, r_1-\omega, r_2)$. If
$$
\psi_{reg1}(\A) \subset(M_{n/2}(\HH)^{\omega} \times M_{n}(\R)^{r_1 -\omega} \times G(\C)^{r_2}),
$$
then
$$
vol(\psi_{abs}(\Lambda))=vol(\psi_{reg1}(\Lambda)) \,\, \mathrm{and} \,\, \delta({\psi_{abs}(\Lambda)})=\delta({\psi_{reg1}(\Lambda)}).
$$
\end{proposition}


\begin{remark}
We note that  the geometric structure of $\psi_{reg1}(\Lambda)$ will in general depend on the choice of $E$-basis of $\A$ and on the choice of the embeddings $\alpha_i$.
\end{remark}

\subsection{Division algebra based $2mn^2$-dimensional codes  in $M_{n\times nm}(\C)$}\label{highdegree}

In the previous section we gave a construction of space time lattice codes from division algebras and described a means of measuring their normalized minimum determinants. We are not yet using the whole signaling space however.  The codes in the previous section are $mn^2$-dimensional lattices in  $M_{n\times nm}(\C)$, while the maximal rank a lattice can have in such a space is $2mn^2$. We now describe a construction of lattices with maximal rank.
The usual strategies for code construction in this scenario can be found in \cite{Lu, YB07}. Unfortunately these methods only allow us to realize some division algebras as lattice codes.  In this section we show how it is possible to overcome these limitations.

Let us consider the  case where the center  $K$ of the  division algebra $\A$ is a totally complex number field. As the center $K$ does not have real primes we simply have an embedding
\begin{equation}
 \A\hookrightarrow G(\C)^{r_2}.
\end{equation}
The space $G(\C)$  consists of pairs of $n\times n$ matrices, where the second matrix is the complex conjugate of the first. Projecting onto the first coordinate gives us an embedding
\begin{equation}
\psi_{abs2}: \A \hookrightarrow M_{n\times n}(\C)^{r_2}.
\end{equation}

\begin{proposition}\label{abs2}
Let $K$ be a totally complex number field of degree $2m$, $\A$ a $K$-central division algebra of degree $n$ and  $\Lambda$ a $\Z$-order in $\A$. Then $\psi_{abs2}(\Lambda)$
is a $2mn^2$-dimensional lattice in $M_{n\times nm}(\C)$ and the following hold:
$$
\mindet{\psi_{abs2}(\Lambda)}= 1, \,\, \,\, m(\psi_{abs2}(\Lambda))= \sqrt{|d(\Lambda/\Z)|\cdot 2^{-mn^2}}
$$
and
$$
\delta(\psi_{abs2}(\Lambda))=\left(\frac{2^{mn^2}}{|d(\Lambda/\Z)|}\right)^{1/4n}.
$$

\end{proposition}
\begin{proof}
The part considering the dimension of the lattice follows directly from Proposition \ref{abs1}. Let us consider the  claim $\mindet{\psi_{abs2}(\Lambda)}= 1$. If we use the mapping $\psi_{abs}$, the  absolute value of the determinant of any codeword  $B$ is given by the formula  $|det(\mathrm{diag}(\psi_{abs}(B)))|=\prod_{1=1}^{2r_2} |b_i|,$ where the $b_i$ are the determinants of $n\times n$ blocks $B_i$ that appear in $B$. However, in this product each $b_i$ can be paired with its complex conjugate. This shows that
$$
|det(\psi_{abs2}(B))|=\sqrt{|det(\psi_{abs}(B))|}\geq 1.
$$
Let us  now consider the Gram matrix of  $\psi_{abs2}(\Lambda)$.  The elements in the matrix are  of type $\Re(tr(\psi_{abs}(a) \psi_{abs}(b)^{\dagger})$. But the relation between mappings $\psi_{abs}$ and $\psi_{abs2}$ already reveals that $\Re(tr(\psi_{abs}(a) \psi_{abs}(b)^{\dagger})=2\Re(tr(\psi_{abs2}(a) \psi_{abs2}(b)^{\dagger})$.
As the Gram matrix is a $2mn^2\times 2mn^2$ matrix we then have that
\begin{align*}
m(\psi(\Lambda))=\sqrt{G(\psi(\Lambda))}=\sqrt{2^{mn^2}G(\psi_{abs2}(\Lambda)})\\=
2^{mn^2/2} m(\psi_{abs2}(\Lambda)).
\end{align*}
The final result now follows from  Lemma \ref{scale}  together with equation  $m(\psi_{abs}(\Lambda))= \sqrt{|d(\Lambda/\Z)|}.$  
\end{proof}

Let us now see how these existence results can be realized  as  explicit codes.

The field $K$ has  $2m$ distinct $\Q$-embeddings $\beta_i: K\hookrightarrow\C$. As we assumed that $K$ is totally complex, each of these embeddings is part of a complex conjugate pair. We will denote by $\bbar{\beta_i}$ the embedding given by  $x\mapsto \bbar{\beta_i(x)}$.

For each $\beta_i$ we can find an embedding $\alpha_i: L\hookrightarrow \C$ such that that $\alpha_i|_{K}=\beta_i$. This choice can be made in such away that
$\bbar{\alpha_i}|_{K}=\bbar{\beta_i}$. Let us now suppose
 $\{\alpha_1,\dots, \alpha_{2m}\}$ is a collection of such embeddings and that they have been ordered in such a way that $\alpha_i=\bbar{\alpha_{i+m}}$, for $0\leq i\leq m$.

With this notation we can now define the following. Let $a$ be an element of $\A$ and $A=\phi(a)$. We then get a mapping $\psi_{reg2}: \A\mapsto M_{n\times nm}(\C)$ by
\begin{equation}\label{main_map}
a\mapsto (\alpha_1(A),\dots, \alpha_{m/2}(A)),
\end{equation}
where each $\alpha_i$  is extended to an embedding $\alpha_i: M_n(E)\hookrightarrow M_n(\C)$.

\begin{proposition}\label{reg2}
Let $\Lambda$ be a $\Z$-order in $\A$ and $\psi_{reg2}$ the previously defined embedding.
Then $\psi_{reg2}(\Lambda)$ is a $2mn^2$-dimensional lattice in $M_{n \times nm}(\C)$ which satisfies
$$
\mindet{\psi_{reg2}(\Lambda)}= 1, \,\, m(\psi_{reg2}(\Lambda))= \sqrt{|d(\Lambda/\Z)|\cdot 2^{-mn^2}}
$$
and
$$
\delta(\psi_{reg2}(\Lambda))=\left(\frac{2^{mn^2}}{|d(\Lambda/\Z)|}\right)^{1/4n}.
$$
\end{proposition}

\begin{remark}
The standard method  to  build multiblock codes with full rate as in \cite{YB07} and \cite{Lu} works only for algebras defined over number fields containing a complex quadratic field. The method described above works for any totally complex center.
\end{remark}

\section{Algebraic and coding theoretic motivation for discriminant bounds}

In the previous section we saw that the normalized coding gain of a code derived from a division algebra depends on the discriminant of the algebra. In the rest of this paper we will concentrate on giving general lower bounds for the discriminants. These will in turn yield upper bounds for the normalized minimum determinants.  Before giving these (purely algebraic) results, let us first examine these bounds and show the manner in which they extend the results of \cite{FOV}.

\subsection{Connection  to the discriminant bounds in number fields}
In \cite{FOV} the authors considered algebraic number field codes in the Rayleigh fast fading SISO channel. 
In our notation the fast fading SISO channel  is simply  a multiblock channel with $n=1$. The  codewords are then of type
$$
(x_1,x_2,\dots, x_m) \in \C^m,
$$
where each of the elements $x_i$ faces an independent fading. 

In order to design a code in this scenario, we can apply the construction of Proposition \ref{reg1}. It calls for a number field $K$ of degree $m$ and a $K$-central division algebra $\A$ of degree $1$; that is, $\A=K$.

Let us now suppose that $\alpha_1,\dots, \alpha_m$ are the $\Q$-embeddings of the field $K$ into $\C$. 
We then have that $\psi_{reg1}(\OO_K)$ is an $m$-dimensional lattice  in $\C^m$. This mapping is the usual Minkowski embedding that has been used in several coding theoretic works.

We can partition the embeddings  $\alpha_1,\dots, \alpha_m$ into $r_1$ real embeddings and $2r_2$ complex embeddings. It follows that
$\psi_{reg_1}(K)\subset \R^{r_1}\times G_1(\C)^{r_2}$. From the basic algebraic number theory we know that $K\otimes_\Q\R  \cong \R^{r_1}\times G_1(\C)^{r_2}$.
According to Propositions \ref{abs1} and \ref{equal} we now have that
\begin{equation}\label{numberfield1}
\delta(\psi_{reg1}(\OO_K))=\frac{1}{\sqrt{|d(K/\Q)|}}.
\end{equation}

In the same manner we may choose a totally complex field $K$ of degree $2m$ so that $\psi_{reg2}(\OO_K))$ will be a $2m$-dimensional lattice in $\C^m$ satisfying
\begin{equation}\label{numberfield2}
\delta(\psi_{reg2}(\OO_K))=\frac{2^{m/2}}{|d(K/\Q|^{1/4}}.
\end{equation}

It is evident that the normalized minimum determinant depends only on the discriminant of the field $K$. In \cite{FOV} the authors then posed the question: What are the limits for the normalized minimum determinant for a given $m$ when one uses these algebraically defined codes?
After all, there are infinitely many isomorphism classes of number fields of each degree $m$. Equations \eqref{numberfield1} and \eqref{numberfield2} transform this problem into finding bounds for discriminants of degree $m$ algebraic number fields.  While multiple number fields may have the same discriminant, it is known that there are only finitely many number fields with a given discriminant. It follows that for every degree $m$ infinitely many discriminants are assumed by degree $m$ number fields. In order to get some intuition for this scenario, the authors of \cite{FOV} used known discriminant bounds of the form described below.

The Odlyzko bound $C_{(r_1, r_2)}$ is a lower bound for the discriminant of all number fields having signature $(r_1, r_2)$. 

As the degree $m\rightarrow\infty$ these bounds give
\begin{equation}\label{odlyzko}
|d(K/\Q)|^{1/m}\geq (60.8)^{r_1/m} (22.3)^{2r_2/m}.
\end{equation}

By employing equations \eqref{numberfield1} and \eqref{numberfield2} the Odlyzko bounds can be transformed into minimum determinant bounds.

\smallskip

We now consider the same question, but in the setting in which we have $n_t$ transmit antennas and employ a Rayleigh block fading channel. The codewords then have form
$$
(X_1,X_2,\cdots, X_m),
$$
where the $X_i$ are $n\times n$ matrices. As we saw earlier, in order to build a code we need degree $m$ number field $K$ (resp. degree $2m$ totally complex number field) and a degree $n$ division algebra. We will then have
$$
\delta({\psi(\Lambda)})=\left(\frac{1}{|d(\Lambda/\Z)|}\right)^{1/2n}\,\,\mathrm{and}\,\, \,\, \delta(\psi_2(\Lambda))=\left(\frac{2^{mn^2}}{|d(\Lambda/\Z)|}\right)^{1/4n}.
$$

Now we can ask the same question as in the case of number fields. If we fix $m$ and $n$, what are the limits for the normalized minimum determinant for codes  in this setting. In the case of number fields the Odlyzko bound immediately implied an upper bound. In the case of division algebras however, the needed bounds do not appear in the mathematical literature. The bounds given in \cite{RoopeThesis} answer to this question only in the case in which the center $K$ is fixed. The bounds in \cite{HL} on the other hand consider only the case of totally complex quartic fields.
 
In this paper we will give completely general lower bounds. Given a center of degree $m$ and a division algebra $\D$ of degree $n$ we will produce lower bounds for the discriminant $d(\Lambda/\Z)$, where $\Lambda$ is any $\Z$-order of $\D$.

\subsection{Scope and implications of the discriminant bounds}

The methods used in the previous sections made use of $\Z$-orders contained in division algebras. We would therefore like to determine lower bounds for the discriminants of these orders. Maximal orders have the smallest discriminant of all $\Z$-orders contained in a given division algebra $\A$. It is therefore sufficient to find lower bounds for the discriminants of maximal orders. This is an enormous help to us as any maximal $\Z$-order contained in a division algebra has an additional integral structure. In particular maximal $\Z$-orders are also $\OO_K$-orders.

\begin{proposition}[{\cite[Theorem 10.5]{Reiner}}]
Let $\A$ be a $K$-central division algebra. Then any maximal $\Z$-order in $\A$ is an $\OO_K$-order.
\end{proposition}
This result will play a crucial role in Section  \ref{bounding_sec}.

Discriminant bounds obviously give bounds for the normalized minimum determinants of the corresponding lattices in the case that we are using construction of Proposition \ref{reg2} or \ref{abs1}.  However, when using Proposition \ref{reg1} the connection between the discriminant and normalized minimum determinant is more subtle. Even in this case however, our bounds are effective.

\smallskip

We also note that the discriminant bounds we give are dependent upon the signature of the algebra, much as the Odlyzko bounds depend on the signature of the number field whose discriminant is being bounded. The need for this dependency is clear as different signatures lead to different codes needed within different coding schemes. If we have a $2$ transmit and receive antennas and we can decode and encode over $2$-blocks of length $2$ without any constraints in decoding complexity then it is a good idea to use the
construction of Proposition \ref{reg2}, which leads to a $16$-dimensional lattice in $M_{2\times 4}(\C)$.  The corresponding discriminant bound is then given by Theorem \ref{theorem:discboundcase1}.  
 
However, if we have the same scenario with only a single receiving antenna and we aim for low decoding complexity, then it is natural to use a code which is an $8$-dimensional lattice in $M_{2\times4}(\C)$. Such code can be naturally be build from the construction of Proposition \ref{reg1}.

The other reason for this division is that, as suggested in \cite{VLL2013}, different signatures seem to lead to considerably different behaviors of the 
inverse determinant sum \eqref{inversesum}. Therefore even two codes having the same center can have very different performances.

\section{Algebraic preliminaries}
\label{sec:algprelim}

Let $K$ be a number field of degree $d$ and signature $(r_1,r_2)$. That is, $d=r_1+2r_2$ where $r_1$ is the number of real embeddings of $K$ and $r_2$ is the number of complex-conjugate pairs of embeddings. Let $\mathcal O_K$ denote the ring of integers of $K$. We impose an order relation on the set of ideals of $\mathcal O_K$ as follows. Given two ideals $I_1$ and $I_2$, we will write $I_1\leq I_2$ if $|N_{\mathbb Q}^K(I_1)|\leq |N_{\mathbb Q}^K(I_2)|$.

Let $\A$ be a central division algebra over $K$ of degree $n$. Given an $\mathcal O_K$-order $\Lambda$ of $A$, we denote by $d(\Lambda/\mathcal O_K)$ the discriminant of $\Lambda$. An order of $\A$ is called {\em maximal} if it is maximal with respect to inclusion. It is well known that all maximal orders of $\A$ have the same discriminant. This quantity is the {\em discriminant of $\A$}.

The following theorem summarizes Theorem 2.4.26 and Proposition 2.4.27 of \cite{RoopeThesis}.

\begin{theorem}\label{theorem:minimaldiscriminanttheorem} Let $\A$ be a central division algebra of degree $n$ over a number field $K$. Let $P_1\leq P_2$ be a pair of prime ideals of $\mathcal O_K$ having smallest norms.
	\begin{enumerate}
	\item If no real place of $K$ ramifies in $\A$ then the discriminant of $\A$ is at least $(P_1P_2)^{n(n-1)}$. 
	\item If $K$ has a unique real place and $n=2m$ with $m$ odd, then the discriminant of $\A$ is at least $P_1^{n(n-1)}P_2^{m(m-1)}$.
	\item If $K$ has at least two real places and $n=2m$ with $m$ odd, then the discriminant of $\A$ is at least $(P_1P_2)^{m(m-1)}$.
	\end{enumerate}
\end{theorem}	
	
\begin{remark}\label{remark:exhaustive} We note that Theorem \ref{theorem:minimaldiscriminanttheorem} is exhaustive in the following sense. The only cases potentially not covered by this theorem are those in which $K$ has no real places or those in which the algebra $\A$ has degree $n=2^km$ over $K$ where $k>1$ and $m$ is odd.  In both of these cases however, one may construct a central division algebra over $K$ of degree $n$ which is unramified at all real places (see \cite[Remark 2.4.24]{RoopeThesis}).
\end{remark}

\section{Bounding the $\mathbb Z$-discriminant of an order}\label{bounding_sec}

Let $\Lambda$ be an $\mathcal O_K$-order of $\A$. The $\mathbb Z$-discriminant of $\Lambda$ is defined by the formula

$$d(\Lambda/\mathbb Z)=N_{K/\mathbb Q}(d(\Lambda/\mathcal O_K))d(\mathcal O_K/\mathbb Z)^{n^2}.$$

\noindent The following theorems provide lower bounds for the $\mathbb Z$-discriminant of $\Lambda$ which depend only on the signatures of $K$ and $\A$. Note that below, $\gamma=0.577215664901532860\dots$ is Euler's constant, and that $C_h$ is the function defined below in Equation (\ref{newfunc}).

Parts (1) - (3) of Theorem \ref{theorem:minimaldiscriminanttheorem} are used to prove Theorems \ref{theorem:discboundcase1} - \ref{theorem:discboundcase3}, respectively.

\begin{theorem}\label{theorem:discboundcase1}
Let $K$ be a number field of degree $d$ and signature $(r_1,r_2)$, $\A$ be a central division algebra over $K$ of degree $n\geq 2$ and signature $(0,r_1,r_2)$, and $\Lambda$ be a maximal order of $\A$. Let $y_0\in\{0.1,2\}$ and $y\leq y_0$ be a positive real number. Lastly, let $z(y) = \left[e^{r_1}e^{d(\gamma+\log 4\pi)}e^{-12\pi/5\sqrt{y}}e^{-I(y)}\right]^{n^2}$ and $(p_1,p_2)$ be the relevant pair of prime powers from Table \ref{table:minvaluecase1table}.
\begin{enumerate}
\item If $y_0=0.1$, then
\[
d(\Lambda/\mathbb Z)\geq\left\{
\begin{array}{lr}
4^{n(n-1)}(53.450)^{n^2}z(y), & n\geq 7\\
(p_1p_2)^{n(n-1)}(e^{C_h(p_1,0.1)+C_h(p_2,0.1)})^{n^2}z(y), & 2\leq n\leq 6
\end{array}
\right.
\]
\item If $y_0=2$, then
\[
d(\Lambda/\mathbb Z)\geq\left\{
\begin{array}{lr}
4^{n(n-1)}(8.134)^{n^2}z(y), & n\geq 7\\
(p_1p_2)^{n(n-1)}(e^{C_h(p_1,2)+C_h(p_2,2)})^{n^2}z(y), & 2\leq n\leq 6
\end{array}
\right.
\]
\end{enumerate}
\end{theorem}

\begin{theorem}\label{theorem:discboundcase2}
Let $K$ be a number field of degree $d$ and signature $(1,r_2)$, $\A$ be a central division algebra over $K$ of degree $n=2m$ (with $m$ odd), and $\Lambda$ be a maximal order of $\A$. Let $y_0\in\{0.1,2\}$ and $y\leq y_0$ be a positive real number. Lastly, let $z(y) = \left[e^{r_1}e^{d(\gamma+\log 4\pi)}e^{-12\pi/5\sqrt{y}}e^{-I(y)}\right]^{n^2}$ and $(p_1,p_2)$ be the relevant pair of prime powers from Table \ref{table:minvaluecase2table}.
\begin{enumerate}
\item If $y_0=0.1$, then
\[
d(\Lambda/\mathbb Z)\geq\left\{
\begin{array}{lr}
2^{n(n-1)}41^{m(m-1)}(9.572)^{n^2}z(y), & n\geq 30\\
p_1^{n(n-1)}p_2^{m(m-1)}(e^{C_h(p_1,0.1)+C_h(p_2,0.1)})^{n^2}z(y), & 2\leq n\leq 26
\end{array}
\right.
\]
\item If $y_0=2$, then
\[
d(\Lambda/\mathbb Z)\geq\left\{
\begin{array}{lr}
2^{n(n-1)}41^{m(m-1)}(2.852)^{n^2}z(y), & n\geq 30\\
p_1^{n(n-1)}p_2^{m(m-1)}(e^{C_h(p_1,2)+C_h(p_2,2)})^{n^2}z(y), & 2\leq n\leq 26
\end{array}
\right.
\]
\end{enumerate}
\end{theorem}

\begin{theorem}\label{theorem:discboundcase3}
Let $K$ be a number field of degree $d$ and signature $(r_1,r_2)$ with $r_1\geq 2$. Let $\A$ be a central division algebra over $K$ of degree $n=2m$ (with $m$ odd), and $\Lambda$ be a maximal order of $\A$. Let $y_0\in\{0.1,2\}$ and $y\leq y_0$ be a positive real number. Lastly, let $z(y) = \left[e^{r_1}e^{d(\gamma+\log 4\pi)}e^{-12\pi/5\sqrt{y}}e^{-I(y)}\right]^{n^2}$ and $(p_1,p_2)$ be the relevant pair of prime powers from Table \ref{table:minvaluecase3table}.
\begin{enumerate}
\item If $y_0=0.1$, then
\[
d(\Lambda/\mathbb Z)\geq\left\{
\begin{array}{lr}
37^{2m(m-1)}(1.803)^{n^2}z(y), & n\geq 118\\
(p_1p_2)^{m(m-1)}(e^{C_h(p_1,0.1)+C_h(p_2,0.1)})^{n^2}z(y), & 6\leq n\leq 114
\end{array}
\right.
\]
\item If $y_0=2$, then
\[
d(\Lambda/\mathbb Z)\geq\left\{
\begin{array}{lr}
9^{2m(m-1)}(1.189)^{n^2}z(y), & n\geq 14\\
(11)^{2m(m-1)}(1.091)^{n^2}z(y), & n=6,10
\end{array}
\right.
\]
\end{enumerate}
\end{theorem}

\begin{remark}
In stating Theorem \ref{theorem:discboundcase3} we excluded the case $n=2$. The reason for this was that in this situation, the hypotheses of the theorem allow for the existence of a quaternion division algebra ramified at precisely two real places of $K$ and which is unramified at all finite primes of $K$. Given a maximal order $\Lambda$ of such an algebra, we will have $d(\Lambda/\mathbb Z)=d(\mathcal O_k/\mathbb Z)$, hence our desired bound is simply the Odlyzko bound.
\end{remark}

\begin{remark} As was the case with Theorem \ref{theorem:minimaldiscriminanttheorem} (and pointed out in Remark \ref{remark:exhaustive}), Theorems \ref{theorem:discboundcase1}, \ref{theorem:discboundcase2} and \ref{theorem:discboundcase3} exhaust all possible central division algebras.
\end{remark}

In order to obtain a lower bound for $d(\Lambda/\mathbb Z)$, it of course suffices to obtain a lower bound for $$|N^K_{\mathbb Q}(d(\Lambda/\mathcal O_K))|^{1/n^2}|d(\mathcal O_K/\mathbb Z)|.$$

We have already seen, in Theorem \ref{theorem:minimaldiscriminanttheorem}, how to obtain lower bounds for $|N_{\mathbb Q}^K(d(\Lambda/\mathcal O_k))|$. We now focus on bounding $|d(\mathcal O_K/\mathbb Z)|$ from below. To do so we will employ the Odlyzko bounds \cite{Odlyzko-bounds}, as well as a refinement of these bounds due to Poitou \cite{Poitou} which takes into account the existence of primes of small norm. The precise formulation of these bounds which we will use is due to Brueggeman and Doud \cite[Theorem 2.4]{doud}.

Let $y>0$ be a real number, $\gamma$ be Euler's constant, and $I(y)$ be as in \cite[Theorem 2.4]{doud}. Let \[
f(x):=(3x^{-3}(\sin x - x\cos x))^2,
\]
and 
\[
C_f(x,y):=4\sum_{j=1}^\infty \frac{\log x}{1+x^j}f(j\sqrt{y}\log x).
\]

Theorem 2.4 of \cite{doud} shows that for any prime ideals $P_1,P_2$ of $k$ and all $y>0$

\begin{equation}\label{equation:odlyzko}
|d(\mathcal O_K/\mathbb Z)| \geq e^{r_1}e^{d(\gamma+\log 4\pi)}e^{-12\pi/5\sqrt{y}}e^{-I(y)}e^{C_f(N^K_{\mathbb Q}(P_1),y)}e^{C_f(N^K_{\mathbb Q}(P_2),y)}.
\end{equation}

We further define functions
$$
h(x) = \left\{
        \begin{array}{ll}
            f(x), & \quad x \leq 4 \\
            0, & \quad x > 4
        \end{array}
    \right.
$$
and
\begin{equation}\label{newfunc}
C_h(x,y):=4\sum_{j=1}^\infty \frac{\log x}{1+x^j}h(j\sqrt{y}\log x).
\end{equation}

The next lemma follows immediately from the fact that for all $x\geq 0$ we have $f(x)\geq h(x)\geq 0$ and the fact that $h(x)$ is decreasing.

\begin{lemma}\label{lemma:basicproperties}
For all real numbers $x,y,y_0>0$ with $y\leq y_0$ the following properties hold:
\begin{enumerate}[(i)]
\item We have $C_f(x,y)\geq C_h(x,y)$.
\item We have $C_h(x,y)\geq C_h(x,y_0)$.
\end{enumerate}
\end{lemma}

It follows that for all $y>0$

\begin{equation}\label{equation:odlyzkomodified}
|d(\mathcal O_K/\mathbb Z)| \geq e^{r_1}e^{d(\gamma+\log 4\pi)}e^{-12\pi/5\sqrt{y}}e^{-I(y)}e^{C_h(N^K_{\mathbb Q}(P_1),y)}e^{C_h(N^K_{\mathbb Q}(P_2),y)}.
\end{equation}

Since we are viewing the signatures of $\A$ and $K$ as being fixed, and since the term $ e^{r_1}e^{d(\gamma+\log 4\pi)}e^{-12\pi/5\sqrt{y}}e^{-I(y)}$ is determined by the signature of $K$, it suffices (by Theorem \ref{theorem:minimaldiscriminanttheorem}) to determine the rational prime powers $p_1,p_2>1$ for which each of the following functions are minimized:

\begin{enumerate}
\item $(p_1p_2)^{(n-1)/n}e^{C_h(p_1,y)}e^{C_h(p_2,y)}=(p_1p_2)^{1-\frac{1}{n}}e^{C_h(p_1,y)}e^{C_h(p_2,y)}$,

\item $p_1^{(n-1)/n}p_2^{m(m-1)/n^2}e^{C_h(p_1,y)}e^{C_h(p_2,y)}=p_1^{1-\frac{1}{n}}p_2^{\frac{1}{4}-\frac{1}{2n}}e^{C_h(p_1,y)}e^{C_h(p_2,y)}$,
\item $(p_1p_2)^{m(m-1)/n^2}e^{C_h(p_1,y)}e^{C_h(p_2,y)}=(p_1p_2)^{\frac{1}{4}-\frac{1}{2n}}e^{C_h(p_1,y)}e^{C_h(p_2,y)}$.
\end{enumerate}

We will determine the minima of these three functions with respect to the parameters $y=0.1$ and $y=2$.

In order to obtain a good bound for $\delta(\Lambda)$, we will take advantage of the fact that both $d(\Lambda/\mathcal O_k)$ and $d(\mathcal O_k/\mathbb Z)$ are affected by the existence of primes of small norm. To do so we will need a few technical results, which are the subject of Section \ref{subsection:minimizing}.

\subsection{Three technical propositions}\label{subsection:minimizing}

\begin{proposition}\label{proposition:theorem1proposition}
Let $n\geq 2$ and define $f_n(x_1,x_2)=(x_1x_2)^{1-\frac{1}{n}}e^{C_h(x_1,y)+C_h(x_2,y)}$. 

\begin{enumerate}
\item If $y=0.1$ and $n\geq 7$ then for all prime powers $p_1,p_2>1$ we have $f_n(p_1,p_2)\geq f_n(2,2)$. For $2\leq n\leq 6$ the prime powers for which $f_n(x_1,x_2)$ is minimized are given in Table \ref{table:minvaluecase1table}.
\item If $y=2$ and $n\geq 7$ then for all prime powers $p_1,p_2>1$ we have $f_n(p_1,p_2)\geq f_n(2,2)$. For $2\leq n\leq 6$ the prime powers for which $f_n(x_1,x_2)$ is minimized are given in Table \ref{table:minvaluecase1table}.
\end{enumerate}
\end{proposition}
\begin{proof}
We will prove the proposition in the case that $y=0.1$. The case in which $y=2$ is completely analogous.

Fix an integer $n\geq 2$ and define an auxiliary function $g(x_1,x_2)=(x_1x_2)e^{C_h(x_1,0.1)+C_h(x_2,0.1)}$. Note that for all $x_1,x_2\geq 0$ we have $g(x_1,x_2)\geq f_n(x_1,x_2)$. An easy calculation shows that $214>g(2,2)$. As $g(2,2)\geq f_n(2,2)$, we conclude that $214\geq f_n(2,2)$. Observe that $f_n(x_1,x_2)\geq \sqrt{x_1x_2}$. It follows that if $p_1,p_2$ are prime powers and $f_n(p_1,p_2)<f_n(2,2)$, then $f_n(p_1,p_2)<214$ and so $2\leq p_1,p_2\leq \frac{214^2}{2}$.

By virtue of the previous paragraph we can check, for any fixed value of $n\geq 2$, to see which values of $(p_1,p_2)$ minimize the function $f_n(x_1,x_2)$ when restricted to prime powers. The assertion of the proposition for $2\leq n\leq 6$ therefore follows immediately. Similarly, an easy computation shows that $f_n(p_1,p_2)\geq f_n(2,2)$ for all prime powers $p_1,p_2$ when $7\leq n\leq 1000$. Suppose now that $n>1000$. Since $f_n(x_1,x_2)=g(x_1,x_2)/(x_1x_2)^{\frac{1}{n}}$, we have $f_n(2,2)> f_n(p_1,p_2)$ if and only if $g(2,2)>(\frac{4}{p_1p_2})^{\frac{1}{n}}g(p_1,p_2)$. As we are assuming that $n>1000$ it is clear that if $(p_1,p_2)\neq (2,2)$ then $(\frac{4}{p_1p_2})^{\frac{1}{n}}>0.990707126780213$. The proposition now follows from a computation which shows that $g(2,2)<0.990707126780213\cdot g(p_1,p_2)$ for all prime powers $p_1,p_2\leq \frac{214^2}{2}$.  \end{proof}

\begin{proposition}
Let $n=2m \geq 2$ with $m$ odd and define $f_n(x_1,x_2)=x_1^{1-\frac{1}{n}}x_2^{\frac{1}{4}-\frac{1}{2n}}e^{C_h(x_1,y)+C_h(x_2,y)}$. 

\begin{enumerate}
\item If $y=0.1$ and $n\geq 30$ then for all prime powers $p_1,p_2>1$ we have $f_n(p_1,p_2)\geq f_n(2,41)$. If $2\leq n\leq 26$, the prime powers for which $f_n(x_1,x_2)$ is minimized are given in Table \ref{table:minvaluecase2table}. 
\item If $y=2$ and $n\geq 14$ then for all prime powers $p_1,p_2>1$ we have $f_n(p_1,p_2)\geq f_n(2,9)$. If $n\in\{2,6,10\}$, the prime powers for which $f_n(x_1,x_2)$ is minimized are $\{(7,17),(3,11),(2,11)\}$. 
\end{enumerate}
\end{proposition}
\begin{proof}
We will prove the proposition in the case that $y=0.1$. The case in which $y=2$ is similar and is left to the reader.

Fix an integer $n\geq 30$ as in the statement of the proposition and define an auxiliary function $g(x_1,x_2)=x_1x_2^{\frac{1}{4}}e^{C_h(x_1,y)+C_h(x_2,y)}$. Then for all $x_1,x_2>0$ we see that $g(x_1,x_2)>f_n(x_1,x_2)$. An easy calculation shows that $49>g(2,41)>f_n(2,41)$. As $n\geq 30$ we see that $49>f_n(2,41)>x_1^{\frac{29}{30}}x_2^{\frac{7}{30}}$. It follows that if $p_1,p_2$ are prime powers for which $f_n(2,41)>f_n(p_1,p_2)$ then $49^{\frac{30}{7}}\geq p_1^{\frac{29}{7}}p_2$. In particular we must have $p_1\leq 47$ and $p_2\leq 992129$.  

By virtue of the previous paragraph we can check, for any fixed value of $n\geq 30$, to see which values of $(p_1,p_2)$ minimize the function $f_n(x_1,x_2)$ when restricted to prime powers. Similarly, an easy computation shows that $f_n(p_1,p_2)\geq f_n(2,41)$ for all prime powers $p_1,p_2$ when $30 \leq n=2m\leq 7000$.

We now assume that $n>7000$. Since $f_n(x_1,x_2)=g(x_1,x_2)/x_1^{\frac{1}{n}}x_2^{\frac{1}{2n}}$, we have $f_n(2,41)> f_n(p_1,p_2)$ if and only if $g(2,41)>(\frac{2}{p_1})^{\frac{1}{n}}(\frac{41}{p_2})^{\frac{1}{2n}}g(p_1,p_2)$. In this case we see that $(\frac{2}{p_1})^{\frac{1}{n}}(\frac{41}{p_2})^{\frac{1}{2n}}\geq (\frac{2}{47})^{\frac{1}{7000}}(\frac{41}{992129
})^{\frac{1}{14000}}=0.998828683870189$ for all prime powers $p_1,p_2$ in the ranges specified above. A computation shows that $g(2,41)<0.998828683870189\cdot g(p_1,p_2)$ for all prime powers $2\leq p_1\leq 47$ and $2\leq p_2\leq 992129$ with $(p_1,p_2)\neq (2,41),(2,37),(2,43)$. The case of the proposition in which $y=0.1$ and $n\geq 30$ now follows from demonstrating that for $n>7000$ and $(p_1,p_2)=(2,37),(2,43)$ we have $f_n(p_1,p_2)\geq f_n(2,41)$. The case in which $y=0.1$ and $2\leq n\leq 26$ can be handled similarly.\end{proof}

\begin{proposition}
Let $n=2m \geq 2$ with $m$ odd and define $f_n(x_1,x_2)=(x_1x_2)^{\frac{1}{4}-\frac{1}{2n}}e^{C_h(x_1,y)+C_h(x_2,y)}$. 

\begin{enumerate}
\item If $y=0.1$ and $n\geq 118$ then for all prime powers $p_1,p_2>1$ we have $f_n(p_1,p_2)\geq f_n(37,37)$. If $6\leq n\leq 114$, the prime powers for which $f_n(x_1,x_2)$ is minimized are given in Table \ref{table:minvaluecase3table}.
\item If $y=2$ and $n\geq 14$ then for all prime powers $p_1,p_2>1$ we have $f_n(p_1,p_2)\geq f_n(9,9)$. If $n=6,10$ then for all prime powers $p_1,p_2>1$ we have $f_n(p_1,p_2)\geq f_n(11,11)$.
\end{enumerate}
\end{proposition}
\begin{proof}
We will prove the proposition in the case that $y=0.1$. The case in which $y=2$ is similar and is left to the reader.

Fix an integer $n\geq 114$ as in the statement of the proposition and define an auxiliary function $g(x_1,x_2)=(x_1x_2)^{\frac{1}{4}}e^{C_h(x_1,y)+C_h(x_2,y)}$. Then for all $x_1,x_2>0$ we see that $g(x_1,x_2)>f_n(x_1,x_2)$. An easy calculation shows that $11>g(37,37)>f_n(37,37)$. As $f_n(x_1,x_2)>(x_1x_2)^{\frac{14}{57}}$ for $n$ in this range, we see that if $p_1,p_2$ are prime powers for which $f_n(p_1,p_2)<f_n(37,37)$ then $(p_1p_2)<11^{\frac{57}{14}}$.

By virtue of the previous paragraph we can check, for any fixed value of $n\geq 114$, to see which values of $(p_1,p_2)$ minimize the function $f_n(x_1,x_2)$ when restricted to prime powers. Similarly, an easy computation shows that $f_n(p_1,p_2)\geq f_n(37,37)$ for all prime powers $p_1,p_2$ when $116 \leq n=2m\leq 20000$. Suppose now that $n>20000$. Since $f_n(x_1,x_2)=g(x_1,x_2)/(x_1x_2)^{\frac{1}{2n}}$, we have $f_n(37,37)> f_n(p_1,p_2)$ if and only if $g(37,37)>(\frac{37}{\sqrt{p_1p_2}})^{\frac{1}{n}}g(p_1,p_2)$. 

Note that $8681$ is the largest prime power less than $\lfloor \frac{11^{\frac{57}{14}}}{2}\rfloor$. As we are assuming that $n>20000$ it is clear that $(\frac{37}{\sqrt{p_1p_2}})^{\frac{1}{n}}\geq (\frac{37}{8681})^{\frac{1}{20000}}=0.999727138528677$ for all prime powers $2\leq p_1,p_2\leq \lfloor\frac{11^{\frac{57}{14}}}{2}\rfloor$. The proof of the $y=0.1, n\geq 114$ case of the proposition now follows from a computation which shows that $g(37,37)<0.999727138528677\cdot g(p_1,p_2)$ for all prime powers $p_1,p_2$ in the aforementioned range. The proof when $y=0.1$ and $6\leq n \leq 110$ is virtually identical.\end{proof}

\newpage

\begin{table}[!h]
\begin{center}
\begin{tabular}[b]{|c|c|c||c|c|c|}
\hline
$y$ & $n$ & $(p_1,p_2)$ & $y$ & $n$ & $(p_1,p_2)$ \\
\hline
$0.1$ & $2$ & $(13,13)$ & $2$ & $2$ & $(7,7)$\\
\hline
$0.1$ & $3$ &  $(7,7)$ & $2$ & $3$ & $(4,4)$\\
\hline
$0.1$ & $4$ & $(4,4)$ & $2$ & $4$ & $(3,3)$ \\
\hline
$0.1$ & $5$ & $(3,3)$ & $2$ & $5$ & $(3,3)$ \\
\hline
$0.1$ & $6$ & $(3,3)$ & $2$ & $6$ & $(3,3)$ \\
\hline

\end{tabular}

\end{center}

\caption{Prime powers $(p_1,p_2)$ for which $(x_1x_2)^{1-\frac{1}{n}}e^{C_h(x_1,y)+C_h(x_2,y)}$ is minimized for $y\in\{0.1,2\}$}
\label{table:minvaluecase1table}
\end{table}

\begin{table}[htbp]
\begin{center}
\begin{tabular}[b]{|c|c|}
\hline
$n$ & $(p_1,p_2)$ \\
\hline
$2$ & $(13,*)$\footnotemark  \\
\hline
$6$ & $(3,64)$ \\
\hline
$10$ &  $(2,53)$ \\
\hline
$14$ & $(2,47)$ \\
\hline
$18$ & $(2,43)$\\
\hline
$22$ & $(2,43)$  \\
\hline
$26$ & $(2,43)$  \\
\hline
\end{tabular}
\end{center}

\caption{Prime powers $(p_1,p_2)$ for which $f_n(x_1,x_2)=x_1^{1-\frac{1}{n}}x_2^{\frac{1}{4}-\frac{1}{2n}}e^{C_h(x_1,0.1)+C_h(x_2,0.1)}$ is minimized}
\label{table:minvaluecase2table}
\end{table}

\begin{table}[htbp]
\begin{center}
\begin{tabular}[b]{|c|c|}
\hline
$n$ & $(p_1,p_2)$ \\
\hline
$6$ & $(64,64)$ \\
\hline
$10$ &  $(53,53)$ \\
\hline
$14$ & $(47,47)$ \\
\hline
$18-26$ & $(43,43)$ \\
\hline
$30-114$ & $(41,41)$ \\
\hline
\end{tabular}
\end{center}

\caption{Prime powers $(p_1,p_2)$ for which $(x_1x_2)^{\frac{1}{4}-\frac{1}{2n}}e^{C_h(x_1,0.1)+C_h(x_2,0.1)}$ is minimized}
\label{table:minvaluecase3table}
\end{table}

\footnotetext{The `*' in Table \ref{table:minvaluecase2table} indicates that when $n=2$ the function $f_n(x_1,x_2)$ does not depend upon $x_2$ and will be minimized whenever $x_1=13$. }
\newpage

\subsection{Proof of Theorems \ref{theorem:discboundcase1}, \ref{theorem:discboundcase2} and \ref{theorem:discboundcase3}} We will now prove Theorem \ref{theorem:discboundcase1}. The proofs of Theorems \ref{theorem:discboundcase2} and \ref{theorem:discboundcase3} are similar and will be left to the reader.

Let $y_0\in\{0.1,2\}$ and $y\leq y_0$ be any positive real number. We have already seen, by combining Theorem \ref{theorem:minimaldiscriminanttheorem}, equation (\ref{equation:odlyzkomodified}) and Lemma \ref{lemma:basicproperties}, that \begin{equation}\label{equation:theorem1minimizer}
d(\Lambda/\mathbb Z)\geq N^K_{\mathbb Q}(P_1P_2)^{n(n-1)}\cdot\left[e^{C_h(N^K_{\mathbb Q}(P_1),y_0)}e^{C_h(N^K_{\mathbb Q}(P_2),y_0)}\right]^{n^2}\cdot \left[e^{r_1}e^{d(\gamma+\log 4\pi)}e^{-12\pi/5\sqrt{y}}e^{-I(y)}\right]^{n^2}.
\end{equation}
We begin by obtaining a lower bound for the related quantity
\begin{equation}\label{equation:theorem1minimizer2}
N^K_{\mathbb Q}(P_1P_2)^{1-\frac{1}{n}}\cdot e^{C_h(N^K_{\mathbb Q}(P_1),y_0)}e^{C_h(N^K_{\mathbb Q}(P_2),y_0)}\cdot e^{r_1}e^{d(\gamma+\log 4\pi)}e^{-12\pi/5\sqrt{y}}e^{-I(y)}
\end{equation} 
\noindent Because we are viewing the signature of $K$ as being fixed, it suffices to simply determine the prime powers $p_1, p_2$ for which $$(p_1p_2)^{1-\frac{1}{n}}e^{C_h(p_1,y_0)+C_h(p_2,y_0)}$$ is minimized. This was done in Proposition \ref{proposition:theorem1proposition}. The theorem follows by substituting these values into (\ref{equation:theorem1minimizer}) and performing simple algebraic manipulations.

\section{A user's guide to discriminant bounds}
\label{examples}
In this section we will discuss how to use the bounds of the previous section and will compare them to certain naive bounds defined below. We give the construction of the naive bound only for the case considered in Theorem \ref{theorem:discboundcase1}, although analogous bounds can be deduced for the other cases a virtually identical manner.

Let $K$ be a number field of degree $d$ and $P_1, P_2$ be the smallest prime ideals of $K$ (with respect to the order relation on the prime ideals of $K$ given in the first paragraph of Section \ref{sec:algprelim}). If we suppose that no infinite place of $K$ is ramified in the degree $n$ central division algebra $\A$ (this is the case when $K$ is totally complex for instance), then for any order $\Lambda\subset\A$ 
we have that by Theorem \ref{theorem:minimaldiscriminanttheorem}
\begin{equation}\label{discriminant2}
d(\Lambda/\Z)\geq(N_{K/\mathbb Q}(P_1) N_{K/\mathbb Q}(P_2))^{n(n-1)} d(\mathcal O_K/\mathbb Z)^{n^2}.
\end{equation}

This equation suggests a trivial bound that can be used to gauge the quality of the bounds proven in the previous section.
Denote by $C_{r_1,d}$ the best known Odlyzko discriminant bound for a degree $d$ number field $K$ containing precisely $r_1$ real primes.
\begin{proposition}\label{trivbound}
Suppose that $K$ is a totally complex number field of degree $d$ and $\A$ is a central division algebra defined over $K$ which has degree $n$. If $\Lambda$ is a $\Z$-order contained in $\A$ then
$$
 d(\Lambda/\mathbb Z)\geq 4^{n(n-1)}(C_{0,d})^{n^2}.
$$
\end{proposition}
\begin{proof}
It is clear that $C_{0,d }\leq d(\mathcal O_K/\mathbb Z)$. As the norm of any prime of $K$ must be at least $2$, the result follows from Equation \eqref{discriminant2}. 
\end{proof}

Let us now see how our bounds in Section \ref{bounding_sec} stack up against this naive bound. In order to compare them, we will  transform the main theorems in Section \ref{bounding_sec} to an easy-to-use form involving classical Odlyzko bounds $C_{r_1,d}$. This is done in Corollaries \ref{cor_6.1}-\ref{cor_6.3}. 


The function $I(y)$ that appeared in  \eqref{equation:odlyzko} depends on the degree $d$ of the field extension $K$ and the number of real embeddings from $K$ into $\R$. More precisely,
$$
I(y)=I_{r_1,d}(y)= \int_{x=0}^{\infty} d \frac{1-f(x\sqrt{y})}{\sinh(x)}+r_1\frac{1-f(x\sqrt{y})}{\cosh(x/2)}dx,
$$
where $d$ is the degree of $K$ and  $r_1$ is the number of real embeddings from $K$ to $\R$. Let $y_{r_1,d}$ the value of $y$ which maximizes 
 \begin{equation}\label{standardOdlyzko}
 e^{r_1}e^{d(\gamma+\log 4\pi)}e^{-12\pi/5\sqrt{y}}e^{-I_{r_1,d}(y)}
\end{equation}
over all real $y>0$. According to \cite{Diaz}, we have 
$$
C_{r_1,d}= e^{r_1}e^{d(\gamma+\log 4\pi)}e^{-12\pi/5\sqrt{y_{r_1,d}}}e^{-I_{r_1,d}(y_{r_1, d})}.
$$

\begin{proposition}\label{assump1}
There exist integers  $1\leq N_1 \leq N_2$ such that  when $d>N_1$ we have that $y_{r_1,d}<2$  and when $d>N_2$ we have $y_{r_1, d}<0.1$.
\end{proposition}

\begin{proof}
Let $y_c$ be a positive real number. We will now prove that when $d$ is large enough the optimal $y$ will be smaller than $y_c$.  Through elementary analysis we   can  see that there exists a positive constant $C$ such that $I_{r_1,d}(y)\geq d C$, for all $y\geq y_c$. Therefore $\frac{12\pi}{5\sqrt{y}} + I_{r_1,d}(y) \geq d C$, for  all $y\geq y_c$.  

It is now enough to prove that there exists such $y$ that
\begin{equation}\label{poit1}
\frac{12\pi}{5\sqrt{y}}+ I_{r_1,d}(y)< d C,
\end{equation}
as in this case  the $y$ must be smaller than $y_c$.

Poitou \cite[p. 6]{Poitou} proves that  for a certain constant $l$ (which is independent of $r_1$ and $d$) we have that 
\begin{equation}\label{poit2}
I_{r_1,d}(y)\leq l y.
\end{equation}

Combining \eqref{poit1}  and  \eqref{poit2} we can see that it  is now enough to prove that when $d$ is large enough we have such  $y$   that
$$
\frac{12\pi}{5}+ d\sqrt{y}(yl-C) < 0,
$$
which, for large enough $d$, is obviously true when $y=C/(l+1)$.\end{proof}

\begin{remark}
This proposition proves that for sufficiently large $d$ our discriminant bounds are effective. Explicitly, calculations in \cite{Diaz} show already that when $d>7$, we have $y_{r_1,d}<2$.
\end{remark}

\begin{remark}
The bounds in \cite{Diaz} are actually calculated by using a  simple approximation of the function $I_{r_1,d}(y)$ (see \cite[p. 16]{Poitou}), which gives slightly weaker bounds. The differences between these weaker bounds and those obtained by optimizing \eqref{standardOdlyzko} are very small and the loss arising from using the tables in  \cite{Diaz} is irrelevant for practical purposes.
\end{remark}

We next state the easy-to-use versions of our bounds in Section \ref{bounding_sec}. Corollaries \ref{cor_6.1}, \ref{cor_6.2}, and \ref{cor_6.3} follow immediately from Proposition \ref{assump1} and Theorems \ref{theorem:discboundcase1}, \ref{theorem:discboundcase2}, and \ref{theorem:discboundcase3}, respectively.

\begin{corollary}
\label{cor_6.1}
Let $K$ be a number field of degree $d$ and signature $(r_1,r_2)$, $\A$ be a central division algebra over $K$ of degree $n\geq 2$ and signature $(0,r_1,r_2)$, and $\Lambda$ be a maximal order of $\A$. Lastly, let $(p_1,p_2)$ be the relevant pair of prime powers from Table \ref{table:minvaluecase1table}.
\begin{enumerate}
\item If $d>N_2$, then
\[
d(\Lambda/\mathbb Z)\geq\left\{
\begin{array}{lr}
4^{n(n-1)}(53.450)^{n^2}(C_{r_1, d})^{n^2}, & n\geq 7\\
(p_1p_2)^{n(n-1)}(e^{C_h(p_1,0.1)+C_h(p_2,0.1)})^{n^2}(C_{r_1, d})^{n^2}, & 2\leq n\leq 6
\end{array}
\right.
\]
\item If $d>N_1$, then
\[
d(\Lambda/\mathbb Z)\geq\left\{
\begin{array}{lr}
4^{n(n-1)}(8.134)^{n^2}(C_{r_1, d})^{n^2}, & n\geq 7\\
(p_1p_2)^{n(n-1)}(e^{C_h(p_1,2)+C_h(p_2,2)})^{n^2}(C_{r_1, d})^{n^2}, & 2\leq n\leq 6
\end{array}
\right.
\]
\end{enumerate}
\end{corollary}

\begin{corollary}
\label{cor_6.2}
Let $K$ be a number field of signature $(1,r_2)$, $\A$ be a central division algebra over $k$ of degree $n=2m$ (with $m$ odd) and $\Lambda$ be a maximal order of $\A$. Lastly, let $(p_1,p_2)$ be the relevant pair of prime powers from Table \ref{table:minvaluecase2table}.

\begin{enumerate}
\item If $d>N_2$, then
\[
d(\Lambda/\mathbb Z)\geq\left\{
\begin{array}{lr}
2^{n(n-1)}41^{m(m-1)}(9.572)^{n^2}(C_{r_1,d})^{n^2}, & n\geq 30\\
p_1^{n(n-1)}p_2^{m(m-1)}(e^{C_h(p_1,0.1)+C_h(p_2,0.1)})^{n^2}(C_{r_1,d})^{n^2}, & 2\leq n\leq 26
\end{array}
\right.
\]
\item If $d>N_1$, then
\[
d(\Lambda/\mathbb Z)\geq\left\{
\begin{array}{lr}
2^{n(n-1)}41^{m(m-1)}(2.852)^{n^2}(C_{r_1,d})^{n^2}, & n\geq 30\\
p_1^{n(n-1)}p_2^{m(m-1)}(e^{C_h(p_1,2)+C_h(p_2,2)})^{n^2}(C_{r_1,d})^{n^2}, & 2\leq n\leq 26
\end{array}
\right.
\]
\end{enumerate}
\end{corollary}

\begin{corollary}
\label{cor_6.3}
Let $K$ be a number field of signature $(r_1,r_2)$ with $r_1\geq 2$, $\A$ be a central division algebra over $k$ of degree $n=2m$ (with $m$ odd) and $\Lambda$ be a maximal order of $\A$. Lastly, let $(p_1,p_2)$ be the relevant pair of prime powers from Table \ref{table:minvaluecase3table}.
\begin{enumerate}
\item If $d>N_2$, then
\[
d(\Lambda/\mathbb Z)\geq\left\{
\begin{array}{lr}
37^{2m(m-1)}(1.803)^{n^2}(C_{r_1,d})^{n^2}, & n\geq 118\\
(p_1p_2)^{m(m-1)}(e^{C_h(p_1,0.1)+C_h(p_2,0.1)})^{n^2}(C_{r_1,d})^{n^2}, & 6\leq n\leq 114
\end{array}
\right.
\]
\item If $d>N_1$, then
\[
d(\Lambda/\mathbb Z)\geq\left\{
\begin{array}{lr}
9^{2m(m-1)}(1.189)^{n^2}(C_{r_1,d})^{n^2}, & n\geq 14\\
(11)^{2m(m-1)}(1.091)^{n^2}(C_{r_1,d})^{n^2}, & n=6,10
\end{array}
\right.
\]
\end{enumerate}
\end{corollary}

We can now immediately see the difference between our bounds and the trivial ones. Both involve $(C_{r_1,d})^{n^2}$, but while the naive bound uses the multiplicative term $4^{n(n-1)}$, we have a considerably larger term. Our bounds are therefore much stronger when the degree $n$  of the algebra is large.

\subsection{Finding optimal algebras}
Through computer simulations we see that when the degree of the center is less than $7$ the value of $y_{r_1, d}$ that maximizes  \eqref{standardOdlyzko} is larger than $2$ and therefore our bounds do not apply. However, for these cases we do not need discriminant bounds as we can simply perform brute force searches for optimal algebras. Let us now explain how these searches can be carried out.

Suppose that $K$ is a totally complex field of degree $d$, and that  $P_1$ and $P_2$ are a pair of smallest prime ideals in $K$. Then there exists a degree $n$ division algebra $\A$ with maximal $\Z$-order $\Lambda$ having discriminant
\begin{equation}\label{discriminant}
d(\Lambda/\Z)=(N_{k/\mathbb Q}(P_1) N_{K/\mathbb Q}(P_2))^{n(n-1)} d(\mathcal O_K/\mathbb Z)^{n^2}.
\end{equation}
Moreover, this is the smallest possible discriminant of a maximal $\Z$-order given that the center of $\A$ is $K$ \cite[Theorem 2.4.26]{RoopeThesis}.  

This formula allows us to perform a brute force search for optimal algebras. The key point is that given a degree $d$ and a real number $M$, there are only finitely many number fields of degree $d$ with discriminant smaller than $M$. We may therefore limit ourselves to the finite search space of degree $d$ fields $K$ with discriminant smaller than 
\begin{equation}\label{bound}
(N_{K/\mathbb Q}(P_1) N_{K/\mathbb Q}(P_2))^{(1-1/n)} d(\mathcal O_K/\mathbb Z),
\end{equation}
and make use of existing tables which contain all number fields of sufficiently small degree and discriminant \cite{LMFD}. For each such field, we find the smallest primes and compute the value of the $\Z$-discriminant given in equation \eqref{discriminant}. We then simply choose the center which minimizes this value.


\begin{example}
Let us demonstrate how this search can be performed in the case of degree $n$ central division algebras defined over a totally complex number field of degree $4$.

When $n=2$ a search through the tables of number fields with signature $(0,2)$ yields a field $K$ of discriminant $d(K)=3^2\cdot 13$ with primitive element having minimal polynomial $x^4-x^3-x^2+x+1$. The field $K$ has primes $P_1$ and $P_2$ both of norm $7$. Hence, there is a degree $2$ division algebra $\A$ containing an order $\Lambda$ such that
$$
d(\Lambda/\Z)=7^4 (3^2\cdot 13)^{4}=449920319121.
$$

We can similarly find the optimal centers for every $n$. The results appear in the following table.

\begin{table}[!h]
\begin{center}
\begin{tabular}[b]{|c|c|c|c|}
\hline
degree of the algebra & $(N_{K/\mathbb Q}(P_1), N_{K/\mathbb Q}(P_2)))$&  $d(\mathcal O_K/\mathbb Z)$& $(d(\Lambda/ \mathbb Z))^{1/n}$\\
\hline
$n=2$ & $(7,7)$ &$3^2\cdot 13$& $49\cdot  3^2\cdot 13^2$  \\
\hline
$n=3$ &  $(4,4)$ &$3^2\cdot 5^2$&$16^2 (3^2\cdot 5^2)^3$   \\
\hline
$n=4$ & $(3,3)$ & $3^2\cdot 37$& $9^3 (3^2\cdot 37)^4  $\\
\hline
$n=5$ & $(3,3)$& $3^2\cdot 37$&  $9^4(3^2\cdot 37)^5$     \\
\hline
$n=6$ & $(3,3)$& $3^2\cdot 37$&  $9^5 \cdot (3^2\cdot 37)^6$  \\
\hline
$n=7$ & $(3,3)$& $3^2\cdot 37$&  $9^6 \cdot (3^2\cdot 37)^7$ \\
\hline
$n>7$ & $(2,2)$& $2^4\cdot 41$&  $4^{n-1}(2^4\cdot 41)^n$ \\
\hline
\end{tabular}
\caption{The optimal algebras with degree  $4$ totally complex centers}
\end{center}
\label{table:optimal}
\end{table}

Here we can see that the optimal center varies as a function of $n$ and the degree of the algebra,  but stabilizes to the field $K$ of discriminant $2^4\cdot 41$ which has two prime ideals with norm 2.  

\end{example}

\begin{remark}
If we use the algebra described in the first line of Table \ref{table:optimal} together with the construction of Proposition \ref{reg2}, we obtain a $16$-dimensional lattice code for the $(2,2,2)$-multiblock channel. 

Previously the best discriminant achieved \cite{HL} corresponded to the center $K$ of discriminant $d(K)=2^4\cdot 3^2$ and primitive element of minimal polynomial $x^4-x^2+1$. The minimal primes in this field have norms $4$ and $9$. The corresponding discriminant therefore is of the form
$$
(4\cdot 9)^2\cdot (2^4\cdot 3^2)^4=557256278016,
$$
revealing that we managed to find an algebra with the smallest known discriminant and therefore also the multiblock code with the largest possible minimum determinant. However, we point out that  in \cite{HL} the authors were concentrating only on fields $K$ that have $\Q(i)$ as a subfield, while we optimized over all totally complex fields.  
\end{remark}

\subsection{Minimum determinant bounds from discriminant bounds}
  
 We conclude the paper by showing how discriminant bounds can be transformed into minimum determinant bounds. As a concrete example, we concentrate on the $(2,2, k)$-multiblock channel.  To apply the construction given in Proposition \ref{reg2}, we need a $d = 2k$-dimensional totally complex field and a degree $2$ division algebra $\A$.  The minimum determinant of any $\Z$-order $\Lambda$ in $\A$ is then given by 
 $$
\delta(\psi_{reg2}(\Lambda))=\left(\frac{2^{4d}}{|d(\Lambda/\Z)|}\right)^{1/8}.
$$ 
By Corollary \ref{cor_6.1}, we have
 $$
d(\Lambda/\mathbb Z)\geq (p_1p_2)^{n(n-1)}(e^{C_h(7,2)+C_h(7,2)})^{n^2} (C'_{r_1, d})^{n^2}\geq (7)^{4}(1.4121)^4 (C'_{0, d})^{4}.
$$

Combining the two previous formulas we see
$$
\delta(\psi_{reg2}(\Lambda))\leq \frac{2^{d/2}}{\sqrt{(9.8847)(C'_{0, d})}}.
$$

According to tables in \cite{Diaz} we find that $ (C'_{0, 8})\geq 5.6^8$ and $(C'_{0, 10})\geq 6.6^{10}$.


In the following table we consider  example algebras and compare these to our bounds. As stated earlier, our bounds are only relevant for degrees $d>7$ and therefore only given in the table below when $d=8$ and $d=10$.
When $d\leq 6$ our example algebras are already optimal. When $d=8$ or $d=10$ the algebras were found through experimentation.

\begin{table}[htbp]
\begin{center}
\begin{tabular}[b]{|c|c|c|c|c|c|c|}
\hline
k&$d$ & $N(p_1) N(p_2)$ & $|d(\mathcal O_k/\mathbb Z)|$ & $\delta(\psi_{reg2}(\Lambda))^{1/d}$ &$(\mathrm{bound})^{(1/d)}$ &Optimality\\
\hline
1&$2$ & $(3,4)$ & $3$ & $ 0.78.. $ & -- &yes\\
\hline
2&$4$ &  $(7,7)$ & $3^2\cdot 13$ & $0.61..$ & -- &yes\\
\hline
3&$6$ & $(13,13)$ & $3^2\cdot 19^2$ & $0.63..$ & -- &yes  \\
\hline
4&$8$ & $(5,9)$ & $5 \cdot 17^{2} \cdot 43^{2} $ & $  0.49.. $ & $\leq 0.52$ &? \\
\hline
5&$10$ & $(11,23)$ & $11^9$ & $0.42..  $ & $\leq 0.50 $&? \\

\hline
\end{tabular}
\end{center}
\end{table}

\section{Acknowledgement}

The third author would like to thank Jyrki Lahtonen for pointing out the short proof of Lemma \ref{spherical}.



\end{document}